\documentclass[10pt]{amsart}
\usepackage{a4}
\usepackage{amssymb}
\usepackage{array}
\usepackage{booktabs}
\usepackage{multirow}
\usepackage{hhline}
\usepackage{color}
\usepackage{nccmath}

\newtheorem{theorem}{Theorem}

\newtheorem{corollary}[theorem]{Corollary}
\newtheorem{lemma}[theorem]{Lemma}
\newtheorem{proposition}[theorem]{Proposition}
\newtheorem{definition}[theorem]{Definition}
\newtheorem{example}{Example}

\newtheorem{remark}[theorem]{Remark}

\numberwithin{theorem}{section}
\numberwithin{equation}{section}

\newcommand{\la}{\langle}
\newcommand{\ra}{\rangle}

\newcommand{\Comp}{\mathbb{C}}

\newcommand{\g}{\mathbb{G}}

\newcommand{\n}{\mathbb{N}}

\newcommand{\tor}{\mathbb{T}}

\newcommand{\z}{\mathbb{Z}}

\begin{document}
\title{Entropic uncertainty relations under localizations on discrete quantum groups}

\author{Sang-Gyun Youn}

\address{Sang-Gyun Youn : Department of Mathematical Sciences, Seoul National University,
San56-1 Shinrim-dong Kwanak-gu, Seoul 151-747, Republic of Korea}
\email{yun87654@snu.ac.kr}

\thanks{ This work is supported by the Basic Science Research Program through the National Research Foundation of Korea (NRF) Grant NRF-2017R1E1A1A03070510 and the National Research Foundation of Korea (NRF) Grant funded by the Korean Government (MSIT) (No.2017R1A5A1015626).}
\begin{abstract}
The uncertainty principle has been established within the framework of locally compact quantum groups in recent years. This paper demonstrates that entropic uncertainty relations can be strengthened under localizations on discrete quantum groups, which is the case if the dual compact quantum group $\g$ is the free orthogonal quantum group $O_N^+$ with $N\geq 3$ or if $\g$ admits an infinite $\Lambda(p)$ set with $p>2$. On the other hand, this paper explains the reason why such phenomena do not appear when $\g$ is one of the connected semisimple compact Lie groups, $O_2^+$ and the quantum $SU(2)$ groups. Also, we discuss the divergence of entropic uncertainty relations together with some explicit explanations.
\end{abstract}

\maketitle

\section{Introduction}

The uncertainty principle has been extensively studied for a long time, especially by Hardy\cite{Ha33}, Hirschman\cite{Hi57}, Beckner\cite{Be75}, Donoho and Stark\cite{DoSt89}. Also, Smith\cite{Sm90}, Tao\cite{Ta05}, {\"O}zaydin and Przebinda\cite{OzPr04} explored the principle in the category of locally compact abelian groups. More recently, Alagic and Russell\cite{AlRu08} established the Donoho-Stark uncertainty principle for a general compact group, and through a series of studies \cite{CrKa14}, \cite{LiWu17} and \cite{JiLiWu17}, Donoho-Stark uncertainty principle, Hirschman-Beckner uncertianty principle and Hardy's uncertainty principle have been studied in the framework of locally compact quantum groups.

Due to \cite{CrKa14}, for $\widehat{\g}$ a discrete Kac algebra, we have
\begin{equation} \label{ineq0}
e^{H(\left |f\right |^2,\varphi_{\g})+H(\left |\widehat{f}\right |^2,\widehat{\varphi}_{\widehat{\g}})}\geq 1
\end{equation}
for all $\widehat{f}\in \ell^1(\widehat{\g})$ with $\left \|\widehat{f}\right\|_{\ell^2(\widehat{\g})}=1$, where $\g$ is the dual compact quantum group of $\widehat{\g}$ and $H(\left |f\right |^2,\varphi_{\g})$ (resp. $H(\left |\widehat{f}\right |^2,\widehat{\varphi}_{\widehat{\g}})$) denotes the relative entropy of $\left |f\right |^2$ (resp. $\left |\widehat{f}\right |^2$) with respect to the Haar state $\varphi_{\g}$ (resp. the left Haar weight $\widehat{\varphi}_{\widehat{\g}}$).

To quote Folland \cite{Fo97}, the uncertainty principle implies that {\it both a non-zero function and its Fourier transform cannot be sharply localized.}  In other words, if a function $f$ is concentrated, then its Fourier transform $\widehat{f}$ should be dispersed.

An interesting fact is that the inequality (\ref{ineq0}) can be improved by utilizing the studies of local Hausdorff-Young inequalities (see \cite{An93}, \cite{Sj95}, \cite{Ka00} and \cite{GaMaPa03}) if $f\in L^2(\tor)$ is highly concentrated and $\widehat{f}\in \ell^1(\z)$. More precisely, we have
\begin{equation}\label{ineq0.1}
\lim_{\epsilon\rightarrow 0}\inf_{\substack{\widehat{f}\in \ell^1(\z):\left \|\widehat{f}\right\|_2=1,\\ supp(f)\subseteq U_{\epsilon}}} e^{H(\left |f\right |^2,\varphi_{\tor})+H(\left |\widehat{f}\right |^2,\widehat{\varphi}_{\z}) } \geq \frac{e}{2}
\end{equation}
where $U_{\epsilon}=\left \{ e^{2\pi i \theta}\in \tor: \left |\theta \right |\leq \epsilon \right\}$. Note that localizations of $f$ on compact group $\tor$ is affected here and see Corollary \ref{cor0.1} for the proof of the inequality (\ref{ineq0.1}). Also, various studies for {\it local uncertainty inequalities} (\cite{Fa78}, \cite{Pr83}, \cite{Pr87}, \cite{PrRa85}, \cite{PrSi88a} and \cite{PrSi88b}) have already noted the localization of $f$ or $\widehat{f}$  in the framework of locally compact groups.

The main view of this paper is that similar, but even much stronger phenomena appear when $\widehat{f}$ is localized on certain discrete quantum groups in the sense that $\mathrm{supp}(\widehat{f})=\left \{\alpha\in \mathrm{Irr}(\g): \widehat{f}(\alpha)\neq 0\right\}$ is contained in a fixed subset $E\subseteq \mathrm{Irr}(\g)$.

The following theorem is a summary of our main results and we will focus only on the spaces $\mathrm{Pol}(\g)=\left \{f\in L^2(\g): \widehat{f}(\alpha)=0~\mathrm{for~all~but~finitely~many~}\alpha \right\}$ and $\mathrm{Pol}_E(\g)=\left \{f\in \mathrm{Pol}(\g):\mathrm{supp}(\widehat{f})\subseteq E \right\}$ in order to simplify some technical issues. Also, the following results for entropic uncertainty relations enhance the Donoho-Stark uncertainty relation due to Proposition \ref{prop1}.

\begin{theorem}\label{thm0}
\begin{enumerate}
\item For the free orthogonal quantum groups $O_N^+$, we have
\begin{equation}
\inf_{\substack{f\in \mathrm{Pol}(\g):\left \|f\right \|_{L^2(\g)}=1,\\ supp(\widehat{f})\subseteq \mathrm{Irr}(O_N^+)\setminus \left \{0 \right \}}} e^{ H(\left |f\right |^2,\varphi_{O_N^+})+H(\left |\widehat{f}\right |^2,\widehat{\varphi}_{\widehat{O_N^+}})} \geq \frac{N}{8}~for~all~N> 8.
\end{equation}
\item For the free orthogonal quantum group $O_N^+$ with $N\geq 3$, we have
\begin{equation}\label{ineq0.2}
\inf_{\substack{f\in \mathrm{Pol}(\g):\left \|f\right \|_{L^2(\g)}=1,\\ supp(\widehat{f})\subseteq \mathrm{Irr}(O_N^+)\setminus \left \{0,1,\cdots, t-1 \right \}}}  e^{H(\left |f\right |^2,\varphi_{O_N^+})+H(\left |\widehat{f}\right |^2,\widehat{\varphi}_{\widehat{O_N^+}})} \gtrsim (\frac{N}{2})^t~for~all~t\in \n.
\end{equation}
\item If a compact Kac algebra $\g$ admits an infinite $\Lambda(p)$-set $E\subseteq \mathrm{Irr}(\g)$ with $p>2$, then
\begin{equation}\label{ineq0.4}
e^{H(\left |f\right |^2,\varphi_{\g})+H(\left |\widehat{f}\right |^2,\widehat{\varphi}_{\widehat{\g}})}\gtrsim \frac{1}{\left \|\widehat{f}\right \|_{\ell^{\infty}(\widehat{\g})}}
\end{equation}
for all $f\in \mathrm{Pol}_E(\g)$ with $\left \|f\right \|_{L^2(\g)}=1$.
\end{enumerate}
\end{theorem}

Note that Theorem \ref{thm0} (1) implies that, even in the slightest support restriction $E=\mathrm{Irr}(O_N^+)\setminus \left \{0\right\}$, the relation (\ref{ineq0}) is strengthened for all $N>8$. Moreover, by Theorem \ref{thm0} (2), we can show
\begin{equation}\label{ineq0.3}
\lim_{t\rightarrow \infty}\inf_{\substack{f\in \mathrm{Pol}(\g):\left \|f\right \|_{L^2(\g)}=1,\\ supp(\widehat{f})\subseteq \mathrm{Irr}(O_N^+)\setminus \left \{0,1,\cdots, t-1 \right \}}}  e^{H(\left |f\right |^2,\varphi_{O_N^+})+H(\left |\widehat{f}\right |^2,\widehat{\varphi}_{\widehat{O_N^+}})} =\infty,
\end{equation}
which can be considered a stronger counterpart of the inequality (\ref{ineq0.1}). However, a phenomenon such as (\ref{ineq0.3}) does not happen
\begin{itemize}
\item if $\widehat{\g}$ is a discrete group,
\item if $\g$ is a connected semisimple compact Lie group,
\item if $\g=O_2^+$ or if $\g=SU_q(2)$ with $0<q<1$.
\end{itemize}
More precisely, in Section \ref{sec:neg}, we will show that 
\begin{equation}\label{ineq0.5}
\sup_{\emptyset \neq E\subseteq \mathrm{Irr}(\g)}  \inf_{f\in \mathrm{Pol}_E(\g):\left \|f\right \|_{L^2(\g)}=1}  e^{H(\left |f\right |^2,\varphi_{\g})+H(\left |\widehat{f}\right |^2,\widehat{\varphi}_{\widehat{\g}})} <\infty
\end{equation} 
for the cases listed above.

On the other hand, Theorem \ref{thm0}, (3) provides a link between the study of lacunarities and uncertainty relations. For example, since the set of generators $\left \{g_j\right\}_{j=1}^{\infty}$ is a Leinert set in $\mathbb{F}_{\infty}$, we are able to obtain the following estimates
\begin{equation}
e^{H(\left | \widehat{f_n}\right |^2, \widehat{\mathbb{F}_{\infty}})+H(\left | f_n \right |^2, \mathbb{F}_{\infty})}\sim n
\end{equation}
for all $f_n=\displaystyle \frac{1}{\sqrt{n}}\sum_{j=1}^n \delta_{g_j}\in \ell^2(\mathbb{F}_{\infty})$ (Refer to Example \ref{ex2}).

In the last section, the divergence of entropic uncertainty relations will be discussed independently. We will show that the uncertainty relations diverge in general. Moreover, the divergence is detectable by establishing the estimates
\begin{equation}
e^{H(\left |\chi_{\alpha} \right |^2,\varphi_{\g})+H(\left |\widehat{\chi_{\alpha}}\right |^2,\widehat{\varphi}_{\widehat{\g}})}\sim n_{\alpha}^2
\end{equation}
when $\g$ is a compact semisimple compact Lie group or $O_N^+$ (Theorem \ref{thm4}). In the cases of quantum $SU(2)$ groups, the divergence does not appear at characters (Proposition \ref{prop3}) whereas appears at certain linear combinations of characters, thanks to the existence of an infinite central $\Lambda(4)$-set (Corollary \ref{thm5}).

\section{Preliminaries}

\subsection{Discrete quantum groups and dual compact quantum groups}

For any discrete quantum group $\widehat{\g}$ there exists a unique compact quantum group $\widehat{\g}=(L^{\infty}(\g),\Delta_{\g},\varphi_{\g})$ such that $L^{\infty}(\g)$ is a von Neumann algebra, $\Delta_{\g}:L^{\infty}(\g)\rightarrow L^{\infty}(\g)\overline{\otimes}L^{\infty}(\g)$ is a normal $*$-homomorphism satisfying
\begin{align*}
(\Delta\otimes \mathrm{id})\circ \Delta= (\mathrm{id}\otimes \Delta)\circ \Delta ,
\end{align*}
and $\varphi_{\g}$ is the unique normal faithful state on $L^{\infty}(\g)$ satisfying
\[(\mathrm{id}\otimes \varphi_{\g})(\Delta_{\g}(a))=\varphi_{\g}(a)1 =(\varphi_{\g} \otimes \mathrm{id})(\Delta_{\g}(a))~\mathrm{for~all~}a\in L^{\infty}(\g).\]
The state $\varphi_{\g}$ is called the {\it Haar state}.

We say that $u=(u_{i,j})_{1\leq i,j\leq n} \in M_n(L^{\infty}(\g))\cong L^{\infty}(\g)\otimes M_n$ is a finite dimensional unitary representation of $\g$ if 
\[u^*u=1\otimes \mathrm{Id}_n =uu^*~\mathrm{and}~\Delta(u_{i,j})=\sum_{k=1}^n u_{i,k}\otimes u_{k,j}~\mathrm{for~all~}1\leq i,j\leq n.\]
Furthermore, if $\left \{T\in M_n: (1\otimes T)u=u(1\otimes T)\right\}=\Comp \cdot \mathrm{Id}_n$, we say that the given unitary representation $u\in M_n(L^{\infty}(\g))$ is irreducible. 

A maximal family of mutually inequivalent (finite dimensional) unitary irreducible representations of $\g$ is denoted by 
\[\mathrm{Irr}(\g) \cong \left \{u^{\alpha}=(u^{\alpha}_{i,j})_{1\leq i,j\leq n_{\alpha}}\in M_{n_{\alpha}}(L^{\infty}(\g))\right\}_{\alpha\in \mathrm{Irr}(\g)}\] and the space of polynomials is defined by
\[\mathrm{Pol}(\g)=\bigoplus_{\alpha \in \mathrm{Irr}(\g)}\mathrm{Pol}_{\alpha}(\g)=\bigoplus_{\alpha \in \mathrm{Irr}(\g)} \mathrm{span}\left \{u^{\alpha}_{i,j}:~ 1\leq i,j\leq n_{\alpha}\right\}\subseteq L^{\infty}(\g).\]

Schur's orthogonality relation says that, for each $\alpha\in \mathrm{Irr}(\g)$, there exists a unique positive invertible matrix $Q_{\alpha}$ (which can be assumed to be diagonal \cite{Da10}) such that
\begin{equation}\label{eq3}
\varphi((u^{\beta}_{s,t})^*u^{\alpha}_{i,j})=\frac{\delta_{\alpha,\beta}\delta_{i,s}\delta_{j,t}(Q_{\alpha})_{i,i}^{-1}}{\mathrm{tr}(Q_{\alpha})}\mathrm{~and~}\varphi(u^{\beta}_{s,t}(u^{\alpha}_{i,j})^*)=\frac{\delta_{\alpha,\beta}\delta_{i,s}\delta_{j,t}(Q_{\alpha})_{j,j}}{\mathrm{tr}(Q_{\alpha})}
\end{equation}
for all $\alpha,\beta\in\mathrm{Irr}(\g)$, $1\leq i,j\leq n_{\alpha}$ and $1\leq s,t\leq n_{\beta}$. We say that $\g$ is of Kac type if the {Haar} state is tracial or equivalently $Q_{\alpha}=\mathrm{Id}_{n_{\alpha}}$ for all $\alpha\in \mathrm{Irr}(\g)$. Also, we define the associated character by $\displaystyle \chi_{\alpha}=\mathrm{tr}(u^{\alpha})=\sum_{j=1}^{n_{\alpha}}u^{\alpha}_{j,j}$ and the quantum dimension by $d_{\alpha}=\mathrm{tr}(Q_{\alpha})$ for each $\alpha\in \mathrm{Irr}(\g)$.

On the dual side, the underlying von Neumann algebra of $\widehat{\g}$ is defined as
\[\ell^{\infty}(\widehat{\g})=\left \{(X(\alpha))_{\alpha\in \mathrm{Irr}(\g)}\in \prod_{\alpha\in \mathrm{Irr}(\g)}M_{n_{\alpha}} : \sup_{\alpha\in \mathrm{Irr}(\g)}\left \|X(\alpha)\right\|<\infty \right\}\]
and the left Haar weight on $\ell^{\infty}(\widehat{\g})$ is described by
\begin{equation}
\widehat{\varphi}_{\widehat{\g}}(X)=\sum_{\alpha\in \mathrm{Irr}(\g)}d_{\alpha} \mathrm{tr}(X(\alpha)Q_{\alpha})
\end{equation}
for all $X=(X(\alpha))_{\alpha}\in \ell^{\infty}(\widehat{\g})_+$.

\subsection{Non-commutative $L^p$-spaces}

For $\g$ a compact quantum group, the von Neumann algebra $L^{\infty}(\g)$ is contractively embedded into $L^1(\g)=L^{\infty}(\g)_*$ under the identification
\[f\mapsto \varphi(\cdot f).\]

We define the non-commutative $L^p$-spaces $L^p(\g)$ to be $(L^{\infty}(\g),L^1(\g))_{\frac{1}{p}}$ where $(\cdot,\cdot)_{\theta}$ denotes the complex interpolation of a compatible pair of Banach spaces. Then the following contractive embeddings 
\[\mathrm{Pol}(\g)\subseteq L^{\infty}(\g)\hookrightarrow L^q(\g) \hookrightarrow L^p(\g)\]
hold for all $1\leq p<q\leq \infty$ and $\mathrm{Pol}(\g)$ is dense in $L^p(\g)$ for all $1\leq p<\infty$. If $\g$ is of Kac type, the non-commutative $L^p$-norm of $a\in L^{\infty}(\g)$ is explicitly given as
\begin{equation}\label{eq2}
\left \|a\right\|_{L^p(\g)}=\varphi(\left |a\right |^p)^{\frac{1}{p}}.
\end{equation}

On the dual side, for $\widehat{\g}$ the discrete dual quantum group, the non-commutative $\ell^p$-spaces are explicitly given by
\begin{equation}\label{eq4}
\ell^{p}(\widehat{\g})=\left \{(A(\alpha))_{\alpha\in \mathrm{Irr}(\g)}\in \ell^{\infty}(\widehat{\g}):\sum_{\alpha\in \mathrm{Irr}(\g)} d_{\alpha} \left \|A(\alpha)Q_{\alpha}^{\frac{1}{p}}\right\|_{S^p_{n_{\alpha}}}^p <\infty \right\}
\end{equation}
for all $1\leq p<\infty$, where $\left \|X\right\|_{S^p_{n_{\alpha}}}$ denotes the Schatten $p$-norm $\mathrm{tr}(\left |X\right |^p)^{\frac{1}{p}}$. Also, the natural $\ell^p$-norm for $A=(A(\alpha))_{\alpha\in \mathrm{Irr}(\g)}\in \ell^p(\widehat{\g})$ is 
\begin{equation}\label{eq5}
\left \|A\right\|_{\ell^p(\widehat{\g})}=(\sum_{\alpha\in \mathrm{Irr}(\g)}d_{\alpha} \left \|A(\alpha)Q_{\alpha}^{\frac{1}{p}}\right\|_{S^p_{n_{\alpha}}}^p)^{\frac{1}{p}}.
\end{equation}

The duality between $\ell^1(\widehat{\g})$ and $\ell^{\infty}(\widehat{\g})$ is given by
\begin{equation}
\la X,A \ra_{\ell^{\infty}(\widehat{\g}),\ell^1(\widehat{\g})}=\widehat{\varphi}_{\widehat{\g}}(XA)=\sum_{\alpha\in \mathrm{Irr}(\g)}d_{\alpha}\mathrm{tr}(X(\alpha)A(\alpha)Q_{\alpha})
\end{equation}
for all $X=(X(\alpha))_{\alpha\in \mathrm{Irr}(\g)}\in \ell^{\infty}(\widehat{\g})$ and $A=(A(\alpha))_{\alpha\in \mathrm{Irr}(\g)}\in \ell^1(\widehat{\g})$. Also, we have $(\ell^{\infty}(\widehat{\g}),\ell^{1}(\widehat{\g}))_{\frac{1}{p}}=\ell^p(\widehat{\g})$ for all $1<p<\infty$. When we deal with the $L^p$-norms, the following complex interpolation theorem will be frequently used.

\begin{theorem}(The interpolation property [Theorem 4.1.2, \cite{BeLo76}])\label{lem3}

Given compatible pairs of Banach spaces $(X_0,X_1)$ and $(Y_0,Y_1)$, if a linear map $T:X_0+X_1\rightarrow Y_0+Y_1$ satisfies $\left \|T\right\|_{X_0\rightarrow Y_0},\left \|T\right\|_{X_1\rightarrow Y_1}<\infty$, then $T:(X_0,X_1)_{\theta}\rightarrow (Y_{0},Y_1)_{\theta}$ is bounded with
\begin{equation}
\left \|T\right\|_{(X_0,X_1)_{\theta}\rightarrow (Y_{0},Y_1)_{\theta}}\leq \left \|T\right\|_{X_0\rightarrow Y_0}^{1-\theta}\left \|T\right\|_{X_1\rightarrow Y_1}^{\theta}
\end{equation} 
for all $0<\theta<1$. In particular, for any $f\in \mathrm{Pol}(\g)$, $X\in M_n$ and $1\leq p_0<p< p_1\leq \infty$ we have
\begin{equation}
\left \|f\right\|_p \leq \left \|f\right \|_{p_0}^{1-\theta} \left \|f\right\|_{p_1}^{\theta}\mathrm{~and~}\left \|XQ^{\frac{1}{p}}\right\|_{S^p_n}\leq \left \|XQ^{\frac{1}{p_0}}\right \|_{S^{p_0}_n}^{1-\theta}\left \|XQ^{\frac{1}{p_1}}\right\|_{S^{p_1}_n}^{\theta},
\end{equation}
where $\displaystyle \frac{1}{p}=\frac{1-\theta}{p_0}+\frac{\theta}{p_1}$ and $Q$ is a positive invertible matrix.
\end{theorem}

\begin{proof}
The last argument follows from the linear maps $T_1:\Comp\rightarrow L^{\infty}(\g), 1\mapsto f$ and $T_2:\Comp\rightarrow M_n, 1\mapsto X$ respectively. The associated compatible pairs are $(L^{p_0}(\g),L^{p_1}(\g))$ and $(L^{p_0}(M_n,Q),L^{p_1}(M_n,Q))$, where $L^p(M_n,Q)$ is a vector space $M_n$ with a norm structure $\left \|A\right \|_{L^p(M_n,Q)}=\mathrm{tr}(\left |AQ^{\frac{1}{p}}\right |^p)^{\frac{1}{p}}$.
\end{proof}

\begin{corollary}\label{cor0}
Let $p\in [1,2)\cup (2,\infty]$ and $f\in \mathrm{Pol}(\g)$ with $\left \|f\right\|_{L^2(\g)}=1$. Then
\begin{equation}
\lim_{q\rightarrow 2}\frac{2q}{2-q}\log(\frac{\left \|f\right \|_{L^q(\g)}}{\left \|\widehat{f}\right \|_{\ell^{q'}(\widehat{\g})}}) \geq \frac{2p}{2-p}\log (\frac{\left \|f\right \|_{L^p(\g)}}{\left \|\widehat{f}\right \|_{\ell^{p'}(\widehat{\g})}}).
\end{equation}
\end{corollary}

\begin{proof}

First of all, the limit exists thanks to [Proposition 5.9, \cite{JiLiWu17}] and [Remark 5.10, \cite{JiLiWu17}]. Let us consider the restriction of the Fourier transform on the one-dimensional subspace $\Comp\cdot f\in L^p(\g)$ and set $T=\mathcal{F}\Bigr|_{\Comp\cdot f}$. If $1\leq p < 2$, then for any $p<q<2$ we have
\[\frac{\left \|\widehat{f}\right \|_{\ell^{q'}(\widehat{\g})}}{\left \|f\right\|_{L^q(\g)}}=\left \|T\right \|_{q\rightarrow q'}\leq \left \|T\right \|_{p\rightarrow p'}^{\frac{(2-q)p}{(2-p)q}}\cdot \left \|T\right\|_{2\rightarrow 2}^{1-\frac{(2-q)p}{(2-p)q}}=(\frac{\left \|\widehat{f}\right \|_{\ell^{p'}(\widehat{\g})}}{\left \|f\right\|_{L^p(\g)}})^{\frac{(2-q)p}{(2-p)q}}\]
by the Plancherel identity (see Subsection \ref{subsec-Fourier}) and Theorem \ref{lem3}.

Therefore, $\displaystyle \frac{2q}{2-q}\log(\frac{\left \|f\right\|_{L^q(\g)}}{\left \|\widehat{f}\right \|_{\ell^{q'}(\widehat{\g})}})\geq \frac{2p}{2-p}\log (\frac{\left \|f\right\|_{L^p(\g)}}{\left \|\widehat{f}\right \|_{\ell^{p'}(\widehat{\g})}})$ and taking limit as $q\nearrow 2$ completes the proof. Also, if $2<p\leq \infty$, then for any $2<q<p$ we have
\[\frac{\left \|f \right \|_{L^q(\g)}}{\left \|\widehat{f}\right \|_{\ell^{q'}(\widehat{\g})}}=\left \|T^{-1}\right \|_{q'\rightarrow q}\leq \left \|T^{-1}\right\|_{p'\rightarrow p}^{\frac{(q-2)p}{(p-2)q}}=(\frac{\left \|f\right\|_{L^p(\g)}}{\left \|\widehat{f}\right \|_{\ell^{p'}(\widehat{\g})}})^{\frac{(q-2)p}{(p-2)q}},\]
so that $\displaystyle \frac{2q}{2-q}\log(\frac{\left \|f\right\|_{L^q(\g)}}{\left \|\widehat{f}\right \|_{\ell^{q'}(\widehat{\g})}})\geq \frac{2p}{2-p}\log(\frac{\left \|f\right \|_{L^p(\g)}}{\left \|\widehat{f}\right \|_{\ell^{p'}(\widehat{\g})}})$, which allows us to obtain the conclusion.
\end{proof}

\begin{corollary}\label{cor0.1}
The inequality (\ref{ineq0.1})
\begin{equation}
\lim_{\delta\rightarrow 0}\inf_{\substack{\widehat{f} \in \ell^1(\z) :\left \|\widehat{f} \right\|_2=1,\\ supp(f)\subseteq U_{\delta}}} e^{H(\left |f\right |^2,\varphi_{\tor})+H(\left |\widehat{f}\right |^2,\widehat{\varphi}_{\z})}\geq \frac{e}{2}
\end{equation}
indeed holds where $U_{\delta}=\left \{e^{2\pi i \theta} \in \tor: \left |\theta\right |\leq \delta \right\}$.
\end{corollary}
\begin{proof} 

Thanks to [Main Theorem, \cite{Sj95}] or [Theorem 2, \cite{Ka00}], for any $\epsilon>0$ and $1<p<2$, there exists $\delta>0$ such that
\[\sup \left \{ \frac{\left \|\widehat{f}\right \|_{\ell^{p'}(\z)}}{\left \|f\right\|_{L^p(\tor)}}:f\in L^p(\tor)\setminus \left \{0\right\}~\mathrm{with~supp}(f)\subseteq U_{\delta} \right\}\leq (1+2\epsilon)\frac{p^{\frac{1}{2p}}}{(p')^{\frac{1}{2p'}}}.\]

 This implies that the map $L^p(U_{\delta})\rightarrow \ell^{p'}(\z), f\mapsto \widehat{f}$ has norm less than $(1+2\epsilon)\displaystyle \frac{p^{\frac{1}{2p}}}{(p')^{\frac{1}{2p'}}}$ when $L^p(U_{\delta})\subseteq L^p(\tor)$ in natural sense. Then, using the fact that $H(\left |f\right |^2,\varphi_{\tor}), H(\left |\widehat{f}\right |^2,\widehat{\varphi}_{\z})<\infty$ whenever $\widehat{f}\in \ell^1(\z)$, [Remark 5.10, \cite{JiLiWu17}] and Corollary \ref{cor0}, we have

\begin{align*}
\lim_{\delta\rightarrow 0}\inf_{\substack{\widehat{f}\in \ell^1(\z):\left \|\widehat{f} \right\|_2=1,\\ supp(f)\subseteq U_{\delta}}} H(\left |f\right |^2,\varphi_{\tor})+H(\left |\widehat{f}\right |^2,\widehat{\varphi}_{\z})&=\lim_{\delta\rightarrow 0}\inf_{\substack{\widehat{f} \in \ell^1(\z) :\left \|\widehat{f} \right\|_2=1,\\ supp(f)\subseteq U_{\delta}}} \lim_{q\rightarrow 2}\frac{2q}{2-q}\log(\frac{\left \|f\right \|_q}{\left \|\widehat{f}\right \|_{q'}})\\
& \geq  \lim_{\delta\rightarrow 0} \inf_{\substack{\widehat{f} \in \ell^1(\z) :\left \|\widehat{f} \right\|_2=1,\\ supp(f)\subseteq U_{\delta}}} \frac{2p}{2-p}\log(\frac{\left \|f\right \|_{L^p(\tor)}}{\left \|\widehat{f}\right \|_{\ell^{p'}(\z)}}) \\
& \geq \frac{2p}{2-p}\log(\frac{(p')^{\frac{1}{2p'}}}{(1+2\epsilon)\cdot p^{\frac{1}{2p}}}).
\end{align*}

Lastly, taking limits as $\displaystyle \epsilon\rightarrow 0$ and $p\nearrow 2$ completes the proof.

\end{proof}

\subsection{Fourier analysis and $\Lambda(p)$ sets}\label{subsec-Fourier}

Let $\g$ be a compact quantum group and $f\in L^{\infty}(\g)$. The sequence of the Fourier coefficients $\widehat{f}=(\widehat{f}(\alpha))_{\alpha\in \mathrm{Irr}(\g)}\in \displaystyle \ell^{\infty}(\widehat{\g})$ is defined by
\[\widehat{f}(\alpha)_{i,j}=\varphi((u^{\alpha}_{j,i})^*f)\]
for all $\alpha\in \mathrm{Irr}(\g)$ and $1\leq i,j\leq n_{\alpha}$. In particular, for any $f\in \mathrm{Pol}(\g)$ we have
\begin{align}
f=\sum_{\alpha\in \mathrm{Irr}(\g)}d_{\alpha} (\widehat{f}(\alpha)Q_{\alpha})_{i,j}u^{\alpha}_{j,i}=\sum_{\alpha\in \mathrm{Irr}(\g)}d_{\alpha} \mathrm{tr}(\widehat{f}(\alpha)Q_{\alpha}u^{\alpha}).
\end{align}
and $\widehat{f}(\alpha)= 0$ for all but finitely many $\alpha$. The map $f\mapsto \widehat{f}$ extends to a contractive map $\mathcal{F}:L^1(\g)\rightarrow \ell^{\infty}(\widehat{\g})$ and an onto isometry $\mathcal{F}_2:L^2(\g)\rightarrow \ell^2(\widehat{\g})$ for each case. In both cases, the maps are called the Fourier transform.

Now, among various notions for lacunarity, let us introduce the $\Lambda(p)$-sets.

\begin{definition}
Let $\g$ be a compact quantum group and $2\leq p<\infty$. We say that $E\subseteq \mathrm{Irr}(\g)$ is a $\Lambda(p)$-set if there exists a universal constant $K=K(E)$ such that
\[\left \|f\right \|_{L^p(\g)}\leq K \left \|f\right\|_{L^2(\g)}\]
for all $f\in \displaystyle \mathrm{Pol}_E(\g)=\left \{f\in \mathrm{Pol}(\g): \widehat{f}(\alpha)=0~\mathrm{for~all~}\alpha\notin E \right\}$.
\end{definition}

\subsection{The entropy and the R$\acute{\bold{e}}$nyi entropy}

In this subsection, we gather some basic notions and properties for entropic quantities. Also, we explain why the entropic uncertainty principle dominates the Donoho-Stark uncertainty principle if $\g$ is of Kac type in Proposition \ref{prop1}.

\begin{definition}
Let $M$ be a von Neumann algebra with $\phi$ a normal semifinite faithful tracial weight.
\begin{enumerate}
\item  For any $f\in M$ such that $\phi(\left |f\right |^p)<\infty$ with $p\in (0,1)\cup (1,\infty)$, the R$\acute{e}$nyi entropy of $f$ with respect to $\phi$ is defined as
\[h_p(f,\phi)=\frac{p}{1-p}\log (\left \|f\right\|_{L^p(M,\phi)})=\frac{1}{1-p}\log(\phi(\left |f\right |^p)).\]
\item For $\rho\in M_+$ such that $\phi(\rho)=1$, the entropy of $\rho\in L^1(M,\phi)_+$ with respect to $\phi$ is defined as
\[H(\rho,\phi)=-\phi(\rho \log(\rho)).\]

\end{enumerate}
\end{definition}

When the entropy of $A\in (M_n)_+$ is discussed with respect to the canonical trace, we simply denote it by $H(A)$.

\begin{lemma}\label{lem2}
\begin{enumerate}
\item Let $\tau$ be a normal faithful finite tracial weight on a von Neumann algebra $M$. Then for any $\rho \in M_+$ with $\tau(\rho)=1$ we have
\begin{equation}
H(\rho,\tau)=\lim_{p\rightarrow 1}\frac{1-\left \|\rho \right \|_{L^p(M,\tau)}}{p-1}=\lim_{p\rightarrow 1}\frac{p}{1-p}\log(\left \| \rho \right\|_{L^p(M,\tau)})=\lim_{p\rightarrow 1} h_p(\rho,\tau) .
\end{equation}

In particular, for any $\xi\in M$ such that $\tau(\xi^*\xi)=1$, we have
\begin{equation}
H(\xi^*\xi , \tau)=\lim_{p\rightarrow 1}\frac{p}{1-p}\log(\left \|\xi^*\xi \right \|_p)=\lim_{p\rightarrow 2}\frac{2p}{2-p}\log (\left \|\xi\right \|_p).
\end{equation}

\item Let $Q\in M_n$ be a positive invertible matrix. Then for any $1<p<\infty$ and $X\in M_n$ we have
\begin{equation}
\frac{\log(\left \|XQ^{\frac{1}{2p}}\right \|_{S^{2p}_n}^{2p})-\log(\left \|XQ^{\frac{1}{2}}\right \|_{HS}^2)}{1-p}\leq \log(\mathrm{tr}(Q))-\log (\left \|XQ^{\frac{1}{2}}\right \|_{HS}^2) .
\end{equation}

\end{enumerate}
\end{lemma}

\begin{proof}
\begin{enumerate}

\item By [Lemma 18, \cite{Te81}], we have 
\[\lim_{p\rightarrow p_0} \frac{h^p-h^{p_0}}{p-p_0}=h^{p_0} \log(h)\]
with respect to the norm topology for $h\in M_+$ and $p_0\in (0,\infty)$. This implies
\[\frac{d}{dp}\Bigr|_{p=1}\left \|\rho \right \|_{L^p(\tau)}=\frac{d}{dp}\Bigr|_{p=1}\left \| \rho\right \|_{L^p(\tau)}^p=\tau(\rho\log(\rho)).\]

\item By the H$\ddot{\mathrm{o}}$lder inequality, we have
\[\left \|Q^{\frac{1}{2p'}}\right \|_{S^{2p'}_n}\left \|XQ^{\frac{1}{2p}}\right \|_{S^{2p}_n}\geq \left \|XQ^{\frac{1}{2}}\right \|_{HS},\]
so that
\begin{align*}
\frac{\log (\left \|XQ^{\frac{1}{2p}}\right \|_{S^{2p}_n}^{2p}) -\log(\left \|XQ^{\frac{1}{2}}\right \|_{HS}^2)}{1-p}&\leq \frac{2p \log(\mathrm{tr}(Q)^{-\frac{1}{2p'}})+ (2p-2)\log (\left \|XQ^{\frac{1}{2}}\right \|_{HS})}{1-p}\\
&=\log(\mathrm{tr}(Q))-\log (\left \|XQ^{\frac{1}{2}}\right \|_{HS}^2)
\end{align*}
\end{enumerate}
\end{proof}

\begin{proposition}\label{prop2}
Suppose $A=(A(\alpha))_{\alpha\in \mathrm{Irr}(\g)}\in \ell^{2}(\widehat{\g})$ satisfies that $\mathrm{supp}(A)$ is finite and $\left \|A\right \|_{\ell^2(\widehat{\g})}=1$. Then the following holds.
\begin{enumerate}
\item We have
\begin{equation}
\lim_{p\rightarrow 2}\frac{2p}{2-p}\log(\left \|A\right \|_{\ell^p(\widehat{\g})})\leq \log(\left |\mathrm{supp}(A)\right |)+2\max_{\alpha\in \mathrm{supp}(A)}\log(d_{\alpha}).
\end{equation}

In particular, if $\g$ is of Kac type, we have

\begin{equation}
\sum_{\alpha\in \mathrm{supp}(A)}n_{\alpha}\log(n_{\alpha})\left \|A(\alpha)\right \|_{HS}^2 \leq H(A^*A, \widehat{\varphi}_{\widehat{\g}})\leq \log(\left | \mathrm{supp}(A)\right |)+2\max_{\alpha\in \mathrm{supp}(A)} \log(n_{\alpha}).
\end{equation}

\item For any $1\leq p\leq 2$ we have
\begin{equation}
\left \|A\right\|_{\ell^p(\widehat{\g})}^p \geq \sum_{\alpha\in S} d_{\alpha}\frac{\left \|A(\alpha)Q_{\alpha}^{\frac{1}{2}}\right \|_{HS}^2}{\left \|A(\alpha)\right \|_{\infty}^{2-p}} \geq \frac{1}{\left \|A\right\|_{\infty}^{2-p}}.
\end{equation}
\end{enumerate}
\end{proposition}

\begin{proof}
\begin{enumerate}

\item First of all, we have
\begin{align*}
&\lim_{p\rightarrow 2}\frac{2p}{2-p}\log(\left \|A\right \|_{\ell^p(\widehat{\g})})=\lim_{t\rightarrow 1} \frac{2t}{1-t} \log (\left \|A \right \|_{\ell^{2t}(\widehat{\g})}) = \lim_{t\rightarrow 1} \frac{\left \|A\right \|_{\ell^{2t}(\widehat{\g})}^{2t}-1}{1-t}\\
&=\sum_{\alpha\in \mathrm{supp}(A)} d_{\alpha} \lim_{t\rightarrow 1}\frac{\left \|A(\alpha)Q_{\alpha}^{\frac{1}{2t}}\right \|_{S^{2t}_{n_{\alpha}}}^{2t}- \left \|A(\alpha)Q_{\alpha}^{\frac{1}{2}}\right \|_{HS}^2 }{1-t}\\
&=\sum_{\alpha\in \mathrm{supp}(A)} d_{\alpha} \left \|A(\alpha)Q_{\alpha}^{\frac{1}{2}}\right \|_{HS}^2 \lim_{t\rightarrow 1}\frac{\log(\left \|A(\alpha)Q_{\alpha}^{\frac{1}{2t}}\right \|_{S^{2t}_{n_{\alpha}}}^{2t})-\log(\left \|A(\alpha)Q_{\alpha}^{\frac{1}{2}}\right \|_{HS}^2)}{1-t}.
\end{align*}
Then, by Lemma \ref{lem2} (2),
\[\leq \sum_{\alpha\in \mathrm{supp}(A)} d_{\alpha} \left \|A(\alpha)Q_{\alpha}^{\frac{1}{2}}\right \|_{HS}^2 (\log(\mathrm{tr}(Q_{\alpha}))-\log(\left \|A(\alpha)Q_{\alpha}^{\frac{1}{2}}\right \|_{HS}^2))\]

For convenience, let us set $d_{\alpha}\left \|A(\alpha)Q_{\alpha}^{\frac{1}{2}}\right\|_{HS}^2=b_{\alpha}$. Then, since $(b_{\alpha})_{\alpha\in \mathrm{supp}(A)}$ is a probability distribution, we have
\begin{align*}
\lim_{p\rightarrow 2}\frac{2p}{2-p}\log(\left \|A\right \|_{\ell^p(\widehat{\g})})& \leq \sum_{\alpha\in \mathrm{supp}(A)} b_{\alpha}\log(\frac{d_{\alpha}^2}{b_{\alpha}}) \\
&\leq \max_{\alpha\in \mathrm{supp}(A)}\log(d_{\alpha}^2)+H((b_{\alpha})_{\alpha\in \mathrm{supp}(A)})\\
&\leq \max_{\alpha\in \mathrm{supp}(A)}\log(d_{\alpha}^2)+ \log (\left | \mathrm{supp}(A)\right |)
\end{align*}

If $\g$ is of Kac type, our assumption implies $\left \|A(\alpha)\right\|_{HS} \leq \displaystyle \frac{1}{\sqrt{n_{\alpha}}}$, so that
\[H(\left |A(\alpha)\right |^2)\geq -\mathrm{tr}(\left |A(\alpha)\right |^2)\log(\mathrm{tr}(\left |A(\alpha)\right |^2))\geq \left \|A(\alpha)\right \|_{HS}^2 \log(n_{\alpha})\]
which leads to our conclusion.

\item Since $\left \|A(\alpha)Q_{\alpha}^{\frac{1}{2}}\right\|_{HS}\leq \left \|A(\alpha)Q_{\alpha}^{\frac{1}{p}}\right\|_{S^p}^{\frac{p}{2}}\left \|A(\alpha)\right\|_{\infty}^{1-\frac{p}{2}}$ by Theorem \ref{lem3}, we have

\begin{align*}
\sum_{\alpha \in S}n_{\alpha}\mathrm{tr}(\left |A(\alpha) Q_{\alpha}^{\frac{1}{p}}\right |^p)&\geq \sum_{\alpha\in S} n_{\alpha}\frac{\left \|A(\alpha)Q_{\alpha}^{\frac{1}{2}}\right \|_{HS}^2}{\left \|A(\alpha)\right \|_{\infty}^{2-p}} \geq \frac{1}{\left \|A\right\|_{\infty}^{2-p}}.
\end{align*}

\end{enumerate}
\end{proof}

\begin{remark}(A rephrasement of [Theorem 5.15, \cite{JiLiWu17}])

Throughout this paper, in order to avoid many notations and explanations for the modular theory, we will describe the entropic uncertainty principle for a compact quantum group $\g$ by
\begin{equation}\label{eq7}
e^{\lim_{q\rightarrow 2}\frac{2q}{2-q}\log(\frac{\left \|f\right \|_{L^q(\g)}}{\left \|\widehat{f}\right\|_{\ell^{q'}(\widehat{\g})}})}\geq 1~\mathrm{for~all~}f\in \mathrm{Pol}(\g).
\end{equation}

Then, from Proposition \ref{prop2} (1) and (\ref{eq7}) , we have the following general bounds:

\begin{equation}\label{eq6}
1\leq e^{\lim_{p\rightarrow 2}\frac{2p}{2-p}\log(\frac{\left \|f\right \|_{L^p(\g)}}{\left \| \widehat{f}\right \|_{\ell^{p'}(\widehat{\g})}})} \leq d_{\alpha}^2
\end{equation}
for all $f\in \mathrm{Pol}_{\alpha}(\g)=\mathrm{Pol}_{\left \{\alpha \right\}}(\g)$.

\end{remark}

The following proposition explains how the study of entropic uncertainty relations dominates the Donoho-Stark uncertainty relations.

\begin{proposition} \label{prop1}
Let $\g$ be a compact quantum group of Kac type. Then for any $f\in \mathrm{Pol}(\g)$ with $\left \|f\right\|_{L^2(\g)}=1$ we have
\begin{equation}
\varphi_{\g}(s(f))\widehat{\varphi}_{\widehat{\g}}(s(\widehat{f}))=\varphi_{\g}(s(f))\sum_{\alpha\in \mathrm{Irr}(\g)}n_{\alpha}\mathrm{rank}(\widehat{f}(\alpha))\geq e^{H(\left |f\right |^2,\widehat{\varphi}_{\g})+H(\left |\widehat{f}\right |^2,\widehat{\varphi}_{\widehat{\g}})},
\end{equation}
where $s(T)$ is the support projection of a bounded operator $T\in B(H)$, i.e. the orthogonal projection onto the closure of range of $T$.
\end{proposition}

\begin{proof}
Since the Haar state $\varphi$ is tracial, we have
\[\left \{ \begin{array}{ll}\left \|f\right \|_{L^p(\g)}=\left \|s(f) \cdot f \right \|_{L^p(\g)}\leq \left \|s(f)\right\|_{\frac{2p}{2-p}}\left \| f\right\|_2\\
\left \|\widehat{f}\right \|_{L^{2}(\g)}= \left \|s(\widehat{f})\cdot \widehat{f}\right \|_{L^{2}(\g)} \leq \left \|s(\widehat{f})\right \|_{\frac{2p}{2-p}}\left \|\widehat{f}\right\|_{p'} \end{array} \right . \]
for all $1\leq p<2$ by the H$\ddot{\mathrm{o}}$lder inequality. Therefore,
\[\varphi(s(f))\geq (\frac{\left \|f\right \|_p}{\left \|f\right\|_2})^{\frac{2p}{2-p}}~\mathrm{and~}\widehat{\varphi}(s(\widehat{f}))\geq (\frac{\left \|\widehat{f}\right\|_2}{\left \|\widehat{f}\right \|_{p'}})^{\frac{2p}{2-p}},\]
so that
\begin{align*}
\log(\varphi(s(\widehat{f}))\widehat{\varphi}(s(\widehat{f}))) \geq \frac{2p}{2-p}\log(\frac{\left \|f \right\|_p}{\left \|\widehat{f}\right\|_{p'}})
\end{align*}
for any $1\leq p<2$. By taking limit as $p\nearrow 2$ and Lemma \ref{lem2} (1), we obtain

\begin{align*}
\log(\varphi(s(f))\sum_{\alpha\in \mathrm{Irr}(\g)}n_{\alpha}\mathrm{rank}(\widehat{f}(\alpha)))&\geq H(\left |f\right |^2,\varphi_{\g})+H(\left |\widehat{f}\right |^2,\widehat{\varphi}_{\widehat{\g}}).
\end{align*}

\end{proof}

\subsection{Examples of compact quantum groups}

\subsubsection{Duals of discrete groups}

Let $\Gamma$ be a discrete group and consider the left translation unitary operators $\lambda_g\in B(\ell^2(\Gamma))$ with $g\in \Gamma$ defined by
\[(\lambda_g \xi)(x)=\xi(g^{-1}x)\]
for all $\xi\in \ell^2(\Gamma)$ and $x\in \Gamma$. It is easy to check $\lambda_g\lambda_h=\lambda_{gh}$ and $\lambda_g^*=\lambda_{g^{-1}}$. The group von Neumann algebra $VN(\Gamma)$ is defined as the closure of a space $\mathrm{span}\left \{\lambda_g:g\in \Gamma\right\}$ in $B(\ell^2(\Gamma))$ with respect to the strong operator topology.

The map $\lambda_g\mapsto \lambda_g\otimes \lambda_g$ uniquely extends to a normal $*$-homomorphism 
\[\Delta:VN(\Gamma)\rightarrow VN(\Gamma)\overline{\otimes} VN(\Gamma).\]

Together with $\Delta$ and the vacuum state $\tau = \la \cdot \delta_e,\delta_e\ra_{\ell^2(\Gamma)}$, $\widehat{\Gamma}=(VN(\Gamma),\Delta,\tau)$ becomes a compact quantum group, which is called {\it the dual of the discrete group $\Gamma$}. In this case, the Haar state $\tau$ of $\widehat{\Gamma}$ is the tracial, $\mathrm{Irr}(\widehat{\Gamma})\cong \Gamma$ and 
\[\mathrm{Pol}(\widehat{\Gamma})=\bigoplus_{g\in \Gamma}\Comp\cdot \lambda_g .\]

\subsubsection{Free orthogonal quantum groups}

Let $N\geq 2$ and $A$ be the universal unital $*$-algebra generated by $\left \{u_{i,j}:1\leq i,j\leq N\right\}$, satisfying the relations
\[u_{i,j}^*=u_{i,j}~\mathrm{and~}\sum_{k=1}^N u_{i,k}u_{j,k}^*=\delta_{i,j}1 = \sum_{k=1}^N u_{k,i}^* u_{k,j}\]
for all $1\leq i,j\leq N$. Then there exists a unital $*$-homomorphism $\Delta_0:A\rightarrow A\otimes A$, $u_{i,j}\mapsto \displaystyle \sum_{k=1}^N u_{i,k}\otimes u_{k,j}$, and a unital faithful positive linear functional $\varphi_0$ on $A$ such that $(\mathrm{id}\otimes \varphi_0)(\Delta_0 (a))=\varphi_0(a)1=(\varphi_0\otimes \mathrm{id})(\Delta_0 (a))$ for all $a\in A$.

The von Neumann algebra $L^{\infty}(O_N^+)$ is defined as the weak$*$-closure of the GNS image of $A$ with respect to the state $\varphi_0$ and then $\Delta_0$ extends to a normal $*$-homomorphism $\Delta_{O_N^+}:L^{\infty}(O_N^+)\rightarrow L^{\infty}(O_N^+)\overline{\otimes}L^{\infty}(O_N^+)$. Also, $\varphi_0$ extends to the Haar state $\varphi_{O_N^+}$ on $L^{\infty}(O_N^+)$.

Together with these structure maps, $O_N^+=(L^{\infty}(O_N^+), \Delta_{O_N^+},\varphi_{O_N^+})$ is a compact quantum group, which is called {\it the free orthogonal quantum group}. In this case, the Haar state is tracial, $\mathrm{Irr}(O_N^+)$ is identified with $\left \{0\right\}\cup \n$ and $n_k\sim r_0^k$ where $r_0=\displaystyle \frac{N+\sqrt{N^2-4}}{2}$ \cite{BaVe09}.

\subsubsection{Quantum $SU(2)$ groups}\label{subsub1}

Let $0<q<1$ and $A$ be the universal unital $C^*$-algebra generated by $a$ and $c$ such that $  \left ( \begin{array}{cc} a&-qc^*\\ c&a^* \end{array}\right )$ is unitary. Then there exists a $*$-homomorphism $\Delta_0:A\rightarrow A\otimes_{min}A, \begin{array}{c}a\mapsto a\otimes a-qc^*\otimes c \\ c\mapsto c\otimes a+a^*\otimes c \end{array}$, and a faithful state $\varphi_0$ on $A$ such that
\[(\mathrm{id}\otimes \varphi_0)(\Delta_0(x))=\varphi_0(x)1=(\varphi_0\otimes \mathrm{id})(\Delta_0(x))\]
for all $x\in A$. Then we define the von Neumann algebra $L^{\infty}(SU_q(2))$ as the weak $*$-closure of GNS image of $A$ with respect to $\varphi_0$.

The maps $\Delta_0$ and $\varphi_0$ uniquely extend to a normal $*$-homomorphism $\Delta_{SU_q(2)}:L^{\infty}(SU_q(2))\rightarrow L^{\infty}(SU_q(2))\overline{\otimes} L^{\infty}(SU_q(2))$ and the unique normal faithful state $\varphi_{SU_q(2)}$ on $L^{\infty}(SU_q(2))$ respectively, which determines {\it the quantum $SU(2)$ group} $SU_q(2)=(L^{\infty}(SU_q(2)),\Delta_{SU_q(2)},\varphi_{SU_q(2)})$. It is known that the Haar state $\varphi_{SU_q(2)}$ is non-tracial, $\mathrm{Irr}(SU_q(2))=\left \{u^n=(u^n_{i,j})_{0\leq i,j\leq n}\in M_{n+1}(L^{\infty}(SU_q(2))):n\in \left \{0\right\}\cup \n \right \}\cong \left \{0\right\}\cup \n$ 
and 
\[Q_n=\left ( \begin{array}{cccc} q^{-n}&\cdots &0&0\\ \vdots &\ddots&\vdots&\vdots\\ 0&\cdots &q^{n-2}&0\\ 0&\cdots &0&q^n  \end{array}\right ) \mathrm{~for~all~}n\in \n \]
with respect to a canonical choice of an orthonormal basis \cite{Ko89}.

\section{The main results}

\subsection{Under the localization on the duals of Free orthogonal quantum groups}

The aim of this section is to demonstrate that the entropic uncertainty relations are drastically sharpened
\[e^{H(\left |f\right |^2,\varphi_{O_N^+})+H(\left |\widehat{f}\right |^2,\widehat{\varphi}_{\widehat{O_N^+}})}>>1\]
if we impose some mild restrictions on the support of $\widehat{f}$ (Theorem \ref{thm2}). The main ingredients are the rapid decay property and the exponential growth of the dual $\widehat{O_N^+}$ of the free orthogonal quantum groups with $N\geq 3$. 

\begin{lemma}([Theorem 4.9, \cite{Ve07}] and [Lemma 3.1, \cite{Br14}]) \label{lem4}

There exists a universal constant $C\leq 2$, which is independent of $N$ and $k$, such that
\begin{equation}
\left \|  f \right\|_{L^{\infty}(O_N^+)}\leq C(1+k)\left \| f\right\|_{L^{2}(O_N^+)}
\end{equation}
for all $N\geq 3$, $k\in \n$ and $f\in \mathrm{Pol}_k(O_N^+)$.
\end{lemma}
\begin{proof}
The proofs of [Corollary 2.3, \cite{Br14}] and [Lemma 3.1, \cite{Br14}] show that the constant can be chosen to satisfy
\[C \leq \frac{1}{1-r_0^{-2}}(\prod_{s=1}^{\infty} \frac{1}{1-r_0^{-2s}})^3.\]

Since $\displaystyle \prod_{s=1}^{\infty}\frac{1}{1-r_0^{-2s}}=\prod_{s=1}^{\infty}(1+\frac{r_0^{-2s}}{1-r_0^{-2s}})\leq e^{\sum_{s=1}^{\infty} \frac{r_0^{-2s}}{1-r_0^{-2s}} }\leq e^{\frac{1}{1-r_0^{-2}}\sum_{s=1}^{\infty}r_0^{-2s}}$, 
\begin{align*}
C&\leq \frac{1}{1-r_0^{-2}}e^{(\frac{1}{r_0-r_0^{-1}})^6}=\frac{1}{1-r_0^{-2}}e^{\frac{1}{(N^2-4)^3}}\leq \frac{1}{1-\frac{4}{N^2}} e^{\frac{1}{125}}\leq \frac{9}{5} e^{\frac{1}{125}}\leq 2.
\end{align*}
\end{proof}

\begin{theorem}\label{thm2}
Let $N\geq 3$, $t\in \n$ and $f\in \mathrm{Pol}_{\mathrm{Irr}(O_N^+)\setminus \left \{0,1,\cdots, t-1\right\}}(O_N^+)$ with $\left \|f\right\|_{L^2(O_N^+)}=1$. Then we have
\begin{equation}
e^{H(\left |f\right |^2,\varphi_{O_N^+})+H(\left |\widehat{f}\right |^2,\widehat{\varphi}_{\widehat{O_N^+}})}\geq \inf_{ k \geq t} \frac{n_{k}}{C^2(1+k)^2}.
\end{equation}
Here, the constant $C$ comes from Lemma $\ref{lem4}$.
\end{theorem}

\begin{proof}
From the assumption and [Proposition 3.7, \cite{Yo18}], we have
\[\sup_{k\geq 0}\frac{n_k^{\frac{1}{2}}\left \|\widehat{f}(k)\right \|_{HS}}{C(k+1)}\leq \left \|f\right\|_{L^1(O_N^+)}\]
for all $f\in L^1(O_N^+)$. This implies
\[\left \|\widehat{f}(k) \right \| \leq \left \|\widehat{f}(k)\right \|_{HS} \leq \frac{C(k+1)}{\sqrt{n_{k}}}\left \|f\right \|_{L^1(O_N^+)}\]
for all $k$ and $f\in L^1(O_N^+)$.

Now, for any $t\in \n$, let us set $E_t=\left \{t,t+1,\cdots \right\}$ and consider linear maps
\[p_t:L^1(O_N^+)\rightarrow \ell^{\infty}_{E_t}(\widehat{O_N^+}), f\mapsto (\widehat{f}( k ))_{ k \geq t},\]
where $\ell^{q}_{E_t}(\widehat{O_N^+})=\left \{A=(A(\alpha))_{\alpha\in \mathrm{Irr}(\g)}\in \ell^q(\widehat{O_N^+}):\mathrm{supp}(A)\subseteq E_t \right\}$ for all $1\leq q\leq \infty$. Then $p_t$ satisfies 
\[ \left \|p_t \right\|_{L^1(O_N^+)\rightarrow \ell^{\infty}_{E_t}(\widehat{O_N^+})}\leq \sup_{k \geq t} \frac{C(1+k)}{\sqrt{n_{k}}}~\mathrm{and~} \left \|p_t\right\|_{L^2(O_N^+)\rightarrow \ell^2_{E_t}(\widehat{O_N^+})}\leq 1 .\]

By Theorem \ref{lem3}, for any $1<p<2$, we obtain
\[\left \| p_t \right \|_{L^p(O_N^+)\rightarrow \ell^{p'}_{E_t}(\widehat{O_N^+})}\leq (\sup_{ k \geq t}\frac{C(1+k)}{n_{ k }^{\frac{1}{2}}})^{\frac{2}{p}-1}.\]

Note that $p_t$ is nothing but the Fourier transform on $\mathrm{Pol}_{E_t}(O_N^+)$. Hence for any $1<p<2$ and $f\in \mathrm{Pol}_{E_t}(O_N^+)$,
\begin{align*}
\frac{2p}{2-p}\log (\frac{\left \|f\right\|_{L^p(O_N^+)}}{\left \|\widehat{f}\right \|_{\ell^{p'}(\widehat{O_N^+})}})&\geq \frac{2p}{2-p}\frac{p-2}{p}\log(\sup_{k\geq t} \frac{C(1+k)}{n_k^{\frac{1}{2}}})\\
&=-2 \log (\sup_{ k \geq t}\frac{C(1+k)}{n_{k}^{\frac{1}{2}}})=2 \inf_{ k \geq t}\log (\frac{n_{ k }^{\frac{1}{2}}}{C(1+ k)}).
\end{align*}

Therefore, by taking limit as $p\nearrow 2$, we have
\begin{align*}
e^{H(\left |f\right |^2,\varphi_{O_N^+})+H(\left |\widehat{f}\right |^2,\widehat{\varphi}_{\widehat{O_N^+}})}&\geq \frac{n_k}{C^2(1+k)^2}.
\end{align*}

\end{proof}

\begin{corollary}\label{cor1}
For the free orthogonal quantum group $O_N^+$ with $N\geq 3$, the following holds.
\begin{enumerate}
\item Under very mild support condition $E=\mathrm{Irr}(O_N^+)\setminus \left \{0\right\}$, we have
\begin{align}\label{ineq5}
\varphi_{O_N^+}(s(f)) \sum_{\alpha\in \mathrm{Irr}(\g)}n_{\alpha}\mathrm{rank}(\widehat{f}(\alpha))&\geq e^{H(\left |f\right |^2,\varphi_{O_N^+})+H(\left |\widehat{f}\right |^2,\widehat{\varphi}_{\widehat{O_N^+}}) } \gtrsim N 
\end{align}
for all $N\in \n$ and $f\in \mathrm{Pol}_{\mathrm{Irr}(O_N^+)\setminus \left \{0\right\}}(O_N^+)$ with $\left \|f\right\|_{L^2(O_N^+)}=1$.

\item  Let $N\geq 3$ and $r<r_0=\displaystyle \frac{N+\sqrt{N^2-4}}{2}$ be fixed. Then we have
\begin{align}
\varphi_{O_N^+}(s(f)) \sum_{\alpha\in \mathrm{Irr}(\g)}n_{\alpha}\mathrm{rank}(\widehat{f}(\alpha))&\geq e^{H(\left |f\right |^2,\varphi_{O_N^+})+H(\left |\widehat{f}\right |^2,\widehat{\varphi}_{\widehat{O_N^+}}) } \gtrsim r^{t}
\end{align}
for any $t\in \n$ and $f\in \mathrm{Pol}_{\mathrm{Irr}(O_N^+)\setminus \left \{0,1,\cdots, t-1\right\}}(O_N^+)$ with $\left \|f\right\|_{L^2(O_N^+)}=1$.

\end{enumerate}
\end{corollary}

\begin{proof}
\begin{enumerate}
\item First of all, the sequence $(\frac{n_k}{(k+1)^2})_{k\geq 1}$ is increasing since
\[\frac{\frac{n_{k+1}}{(k+2)^2}}{\frac{n_k}{(k+1)^2}}=\frac{n_{k+1}}{n_k}\cdot (\frac{k+1}{k+2})^2\geq \frac{N+\sqrt{N^2-4}}{2}\cdot (\frac{k+1}{k+2})^2\geq \frac{3+\sqrt{5}}{2}\cdot \frac{4}{9}>1.\]

Thus, the conclusion follows from Theorem \ref{thm2}, Proposition \ref{prop1} and $\displaystyle \frac{n_1}{4C^2}=\frac{N}{4C^2} \geq \frac{N}{8}$.
\item Again, it is enough to use Theorem \ref{thm2}, Proposition \ref{prop1} and the fact that
\[ \inf_{k\geq t}\frac{n_k}{C^2(1+k)^2}\geq \inf_{k\geq t}\frac{r_0^k}{4 (1+k)^2}=\frac{r_0^t}{4(1+t)^2}\gtrsim r^t~\mathrm{for~all~}t\in \n.\]
\end{enumerate}
\end{proof}

\begin{remark}
In fact, the proof of Theorem \ref{thm2} is still valid for other various free quantum groups. A similar conclusion holds for 
\begin{itemize}
\item the free unitary quantum groups $U_N^+$ with $N\geq 3$ and
\item the quantum automorphism groups $\g_{aut}(B,\psi)$ with a $\delta$-trace $\psi$ and a $C^*$-algebra $B$ such that $\mathrm{dim}(B)\geq 5$,
\end{itemize}
if we choose $E_t=\left \{\alpha\in \mathrm{Irr}(\g):\left |\alpha\right |\geq t\right\}$ where $\left |\cdot  \right |$ denotes the natural length function on $\mathrm{Irr}(\g)$.
\end{remark}

\subsection{When $\widehat{f}$ is localized on a lacunary set}

One of the main observations we obtained is that studies on {\it lacunarity} can be used to estimate  uncertainty relations between $f$ and $\widehat{f}$. In particular, we will explore how the existence of infinite $\Lambda(p)$ sets affect the uncertainty relations. It is known that all duals of discrete groups, which include all of the abelian compat groups, admits an infinite $\Lambda(p)$-set [Theorem A.1., \cite{Wa17}]. For another link between lacunarity and uncertainty principles, refer to \cite{NaSi11}.

\begin{theorem}\label{thm3}
If $E\subseteq \mathrm{Irr}(\g)$ is a $\Lambda(p)$-set with a constant $K>0$ and $p>2$, then 
\begin{equation} 
e^{\lim_{q\rightarrow 2}\frac{2q}{2-q}\log(\frac{\left \|f\right \|_q}{\left \|\widehat{f}\right \|_{q'}})}\geq \frac{1}{K^{\frac{2p}{p-2}}\left \|\widehat{f}\right\|_{\infty}^2}  
\end{equation}
for all $f\in \displaystyle \mathrm{Pol}_E(\g)$ with $\left \|f\right\|_2=1$. In particular, if $\g$ is of Kac type, then
\begin{equation}
\varphi_{\g}(s(f)) \sum_{\alpha\in \mathrm{Irr}(\g)}n_{\alpha}\mathrm{rank}(\widehat{f}(\alpha))\geq e^{H(\left |f\right |^2,\varphi_{\g})+H(\left |\widehat{f}\right |^2,\widehat{\varphi}_{\widehat{\g}})}\gtrsim \frac{1}{\left \|\widehat{f}\right\|_{\infty}^2} \geq \min_{\alpha\in \mathrm{supp}(\widehat{f})}n_{\alpha} 
\end{equation}
for all $f\in \mathrm{Pol}_E(\g)$ with $\left \|f\right \|_2=1$.
\end{theorem}

\begin{proof}
By Theorem \ref{lem3} and Proposition \ref{prop2}, for all $2<q<p$ and $f\in \mathrm{Pol}_E(\g)$ with $\left \|f\right\|_2=1$, we have
\[\left \|f\right\|_q \leq K^{\frac{p(q-2)}{q(p-2)}}\left \|f\right\|_2 \]
and
\begin{align*}
\frac{2q}{q-2}\log (\frac{\left \|\widehat{f}\right \|_{q'}}{\left \|f\right \|_q})& = \frac{2q}{q-2}\log (\frac{(\sum_{\alpha\in \mathrm{supp}(\widehat{f})}d_{\alpha}\left \|\widehat{f}(\alpha)Q_{\alpha}^{\frac{1}{q'}}\right\|_{S^{q'}_{n_{\alpha}}}^{q'})^{\frac{1}{q'}}}{\left \|f\right \|_q}) \\
&\geq \frac{2q}{q-2}\log(\frac{(\sum_{\alpha\in \mathrm{supp}(\widehat{f}) }d_{\alpha}\frac{\left \|\widehat{f}(\alpha)Q_{\alpha}^{\frac{1}{2}}\right \|_2^2}{\left \|\widehat{f}(\alpha)\right \|_{\infty}^{2-q'}})^{\frac{1}{q'}}}{K^{\frac{p(q-2)}{q(p-2)}}})\\
& \geq \frac{2q}{q-2}\log(\frac{\left \|\widehat{f}\right\|_{\infty}^{\frac{2-q}{q}}}{K^{\frac{p(q-2)}{q(p-2)}}}) =\log(\frac{1}{\left \|\widehat{f}\right \|_{\infty}^2 K^{\frac{2p}{p-2}}}).
\end{align*}

Hence, we obtain
\begin{equation}\label{ineq11}
e^{\lim_{q\rightarrow 2} \frac{2q}{2-q}\log(\frac{\left \|f\right \|_q}{\left \|\widehat{f}\right \|_{q'}}) } \geq  \frac{1}{\left \|\widehat{f}\right\|_{\infty}^2K^{\frac{2p}{p-2}}}.
\end{equation}

If $\g$ is of Kac type, the above (\ref{ineq11}) becomes
\[e^{H(\left |f\right |^2,\varphi_{\g})+H(\left |\widehat{f}\right |^2,\widehat{\varphi}_{\widehat{\g}})}\gtrsim \frac{1}{\left \|\widehat{f}\right\|_{\infty}^2}.\]

Also, Proposition \ref{prop1} and the fact that $\displaystyle \left \|\widehat{f}(\alpha)\right \|\leq \left \|\widehat{f}(\alpha)\right\|_{HS}\leq \frac{1}{\sqrt{n_{\alpha}}}$ for all $\alpha\in \mathrm{Irr}(\g)$ demonstrate the last argument.

\end{proof}

From now on, let us gather some well-known explicit examples of infinite $\Lambda(p)$-sets. On typical examples $\tor$ (among commutative ones) and $\widehat{\mathbb{F}_{\infty}}$ (among co-commutative ones), we introduce Hadamard sets and a Leinert set below.

\begin{example} \label{ex3} (Hadamard sets)

It is known that any Hadamard set $E=\left \{n_j\right\}_{j=1}^{\infty}\subseteq \n$ with $\displaystyle \inf_{j\in \n}\frac{n_{j+1}}{n_j}>1$ is a Sidon set (\cite{Si27a} and \cite{Si27b}), which automatically becomes a $\Lambda(p)$ set for all $2<p<\infty$ [Theorem 6.3.9, \cite{GrHa13}]. Hence, we have
\begin{equation}\label{ineq8}
e^{H(\left |f\right |^2,\varphi_{\tor})+H(\left |\widehat{f}\right |^2,\widehat{\varphi}_{\z}) }\gtrsim \frac{1}{\left \|\widehat{f}\right\|_{\infty}^2}
\end{equation}
for all $f \in \mathrm{Pol}_E(\tor)$ with $\left \|f\right\|_2=1$. In particular, we have
\begin{equation}
e^{H(\left |f_m\right |^2,\varphi_{\tor})+H(\left |\widehat{f_m}\right |^2,\widehat{\varphi}_{\z}) }\sim m
\end{equation}
for all $f_m\sim \displaystyle \frac{1}{\sqrt{m}}\sum_{j=1}^{m}z^{n_j}$ by Proposition \ref{prop2} and (\ref{ineq8}).
\end{example}

\begin{example} \label{ex2} (A Leinert set in $\mathbb{F}_{\infty}$ \cite{Le74}, \cite{Bo75} )

Let $E=\left \{g_j\right\}_{j=1}^{\infty}\subseteq \mathbb{F}_{\infty}$ be the generators of the free group $\mathbb{F}_{\infty}$. Since
\begin{equation}
\left \| \sum_{j=1}^n a_j\lambda_{g_j}\right\|_{VN(\mathbb{F}_{\infty})}\leq 2(\sum_{j=1}^n \left |a_j\right |^2)^{\frac{1}{2}}
\end{equation}
for any $n\in \n$ and scalars $a_j$, the subset $E$ is a $\Lambda(p)$-set with a universal constant $K=2$ for all $2<p<\infty$. Therefore, Theorem \ref{thm3} and taking limit as $p\rightarrow \infty$ gives us

\begin{equation}\label{ineq9}
e^{H(\left |f\right |^2,\varphi_{\widehat{\mathbb{F}_{\infty}}})+H(\left |\widehat{f}\right |^2,\widehat{\varphi}_{\mathbb{F}_{\infty}})}\geq \frac{1}{4\left \|\widehat{f}\right \|_{\infty}^2}
\end{equation}
for all $f\in \mathrm{Pol}_E(\widehat{\mathbb{F}_{\infty}})$ with $\left \|f\right\|_2=1$. In particular, by Proposition \ref{prop2} and (\ref{ineq9}), we have
\begin{equation}
\frac{n}{4}\leq e^{H(\left | \frac{1}{\sqrt{n}}\sum_{j=1}^n \lambda_{g_j}\right |^2, \varphi_{\widehat{\mathbb{F}_{\infty}}} )+H(\left |\frac{1}{\sqrt{n}}\sum_{j=1}^n \delta_{g_j}\right |^2,  \widehat{\varphi}_{\mathbb{F}_{\infty}} )}\leq n ~\mathrm{for~all~}n\in \n.
\end{equation}
\end{example}

A well-known example of infinite $\Lambda(p)$-sets $E$ with unbounded degrees, i.e. $\displaystyle \sup_{\alpha\in E}n_{\alpha}=\infty$, is so-called a $FTR$ set. Refer to \cite{GrHa13} for details.

\begin{example}[A FTR set]\label{ex1}

Suppose that $G=\displaystyle \prod_{j=1}^{\infty}U(2^j)$ and let $\pi_j:G\rightarrow U(2^j)$ be the canonical projections. Then $E=\left \{\pi_j\right\}_{j = 1}^{\infty}\subseteq \mathrm{Irr}(G)$ is a Sidon set, which is a $\Lambda(p)$ set for all $2<p<\infty$. Therefore, we have

\begin{equation}
e^{H(\left |f\right |^2,\varphi_G)+H(\left |\widehat{f}\right |^2,\widehat{\varphi}_{\widehat{G}})} \gtrsim  2^n 
\end{equation}
for all $n\in \n$ and $f\in \mathrm{Pol}_{E\setminus \left \{\pi_1,\cdots, \pi_n\right\}}(G)$ with $\left \|f\right\|_{L^2(G)}=1$.

\end{example}

\section{On the duals of compact Lie groups, $O_2^+$ and $SU_q(2)$}\label{sec:neg}

In contrast to the previous section, the explored improvements of uncertainty relations (Corollary \ref{cor1} (2) or Theorem \ref{thm3}) do not appear in the cases listed below. 
\begin{itemize}
\item $\g$ is a compact semisimple connected Lie group,
\item $\g=O_2^+$ or $\g=SU_q(2)$ with $0<q<1$. 
\end{itemize}

First of all, it is known that all connected semisimple compact Lie groups and $SU_q(2)$ do not admit infinite (local) $\Lambda(p)$-sets for any $p>2$ (\cite{GiTr80} and [Proposition 5.17, \cite{Wa17}]). Also, through the proof of Theorem \ref{thm1}, there exist no infinite $\Lambda(p)$-sets for $O_2^+$. These facts imply that Theorem \ref{thm3} is no longer applicable.

From now on, we will focus on the validity of Corollary \ref{cor1} (2) for the cases above. Recall that Corollary \ref{cor1} (2) or Example \ref{ex1} provides us with a sequence of finite subsets $(E_t)_{t\in \n}\subseteq \mathrm{Irr}(\g)$ such that
\[\lim_{t\rightarrow \infty} \inf_{f\in \mathrm{Pol}_{\mathrm{Irr}(\g)\setminus E_t}(\g):\left \|f\right\|_2=1}e^{\lim_{p\rightarrow 2}\frac{2p}{2-p}\log(\frac{\left \|f\right\|_p}{\left \|\widehat{f}\right \|_{p'}})}=\infty.\]

However, we will show that such a sequence does not exist if $\g$ is a connected semisimple compact Lie group, $O_2^+$ or $SU_q(2)$ with $0<q<1$. More precisely, for the compact quantum groups listed above, we will prove that
\begin{equation}
\sup_{\emptyset \neq E\subseteq \mathrm{Irr}(\g)} \inf_{f\in \mathrm{Pol}_E(\g):\left \|f\right\|_2=1}e^{\lim_{p\rightarrow 2}\frac{2p}{2-p}\log(\frac{\left \|f\right \|_p}{\left \|\widehat{f}\right \|_{p'}})}<\infty .
\end{equation}

Our strategy is to look closely at $L^p$-norms of certain matrix elements $u^{\alpha}_{i,i}$.

\begin{lemma}\label{lem1}
\begin{enumerate}
\item (Main theorem, \cite{GiTr80})

Let $G$ be a connected semisimple compact Lie group. Then for any $p>0$, there exist universal constants $A_p$ and $B_p$ such that for any unitary irreducible representation $\pi\in \mathrm{Irr}(G)$, there exists a matrix coefficient $a_{\pi}=\la \pi(\cdot)\xi,\xi \ra_{H_{\pi}}$ with $\left \|\xi\right\|_{H_{\pi}}=1$ such that
\[A_p <   n_{\pi} \left \|a_{\pi}\right\|_{p}^p<B_p.\]
\item For any $n\in \left \{0\right\}\cup \n$, $0\leq i,j\leq n$ and $1\leq p<\infty$, we have
\[\frac{\left \|\widehat{u^n_{i,j}}\right \|_{\ell^p(\widehat{SU_q(2)})}}{\left \|\widehat{u^n_{i,j}}\right \|_{\ell^2(\widehat{SU_q(2)})}}=(d_n(Q_n)_{i,i})^{\frac{2-p}{2p}}.\]
\end{enumerate}
\end{lemma}

\begin{proof}
(2) From the explicit formulas for $\ell^p$-norms (\ref{eq5}) and the fact that $\displaystyle \widehat{u^n_{i,j}}(k)=\frac{\delta_{n,k}(Q_n)_{i,i}^{-1}}{d_n}E^n_{j,i}$, we have
$\displaystyle \frac{\left \|\widehat{u^n_{i,j}}\right \|_{\ell^p(\widehat{SU_q(2)})}}{\left \|\widehat{u^n_{i,j}}\right \|_{\ell^2(\widehat{SU_q(2)})}}=(d_n(Q_n)_{i,i})^{\frac{2-p}{2p}}$
for all $1\leq p<\infty$.
\end{proof}

\begin{theorem}\label{thm1}

Let $\g$ be one of connected semisimple compact Lie groups, the free orthogonal quantum group $O_2^+$ and the quantum $SU(2)$ group $SU_q(2)$ with $0<q<1$. Then we have
\begin{equation}
\sup_{\emptyset \neq E\subseteq \mathrm{Irr}(\g)}\inf_{f\in \mathrm{Pol}_E(\g):\left \|f\right\|_{L^2(\g)}=1} e^{\lim_{p\rightarrow 2}\frac{2p}{2-p}\log(\frac{\left \|f\right \|_{L^p(\g)}}{\left \|\widehat{f}\right \|_{\ell^{p'}(\widehat{\g})}})}<\infty.
\end{equation}

\end{theorem}

\begin{proof}
\begin{enumerate}
\item (The case of connected semisimple compact Lie groups)

For each $\pi\in \mathrm{Irr}(G)$, let us choose $a_{\pi}$ such that $\left \|a_{\pi}\right\|_{L^1(G)}\leq B_1 n_{\pi}^{-1}$ from Lemma \ref{lem1} (1). Then, by Theorem \ref{lem3}, we have $\left \|a_{\pi}\right\|_{L^p(G)} \leq B_1^{\frac{2-p}{p}}n_{\pi}^{-\frac{1}{p}}$ for all $\pi\in \mathrm{Irr}(G)$ and $1<p<2$.
Therefore,
\begin{align*}
\frac{2p}{2-p}\log(\frac{\left \|\sqrt{n_{\pi}}a_{\pi}\right \|_p}{\left \|\sqrt{n_{\pi}}\widehat{a_{\pi}}\right \|_{p'}})&=\frac{2p}{2-p}\log(\frac{\left \|\sqrt{n_{\pi}}a_{\pi}\right \|_p}{n_{\pi}^{\frac{1}{2}-\frac{1}{p}}})\\
&\leq \frac{2p}{2-p}\log(B_1^{\frac{2-p}{p}} )=\log(B_1^2),
\end{align*}
so that
\begin{align*}
&e^{H(\left |\sqrt{n_{\pi}}a_{\pi}\right |^2,\varphi_G)+H(\left |\sqrt{n_{\pi}}\widehat{a_{\pi}}\right |^2,\widehat{\varphi}_G)} \leq B_1^2.
\end{align*}

\item (The case of $O_2^+$)

The idea is to transfer some explicit computations on $SU(2)$ into the situation of $O_2^+$. First of all, in the case of $SU(2)$, we have $\mathrm{Irr}(SU(2))=\left \{\pi^n=(\pi^n_{i,j})_{0\leq i,j\leq n}\right\} \cong \left \{0\right\}\cup \n$ and $\displaystyle \left \|\pi^n_{0,0}\right\|_p=\frac{1}{(\frac{np}{2}+1)^{\frac{1}{p}}}$ for all $p$.

On the other side, $O_2^+\cong SU_{-1}(2)$ as compact quantum groups and $\mathrm{Irr}(SU_{-1}(2))=\left \{u^n=(u^n_{i,j})_{0\leq i,j\leq n}:n\in \left \{0\right\}\cup \n \right\}\cong \left \{0\right\}\cup \n$. If we denote by $a=u^1_{0,0}$, then $a^*a=aa^*$,  $u^n_{0,0}=a^n$ thanks to [Theorem 5.4, \cite{Ko89}].  Therefore, the moments of $(a^*a)^n=(u^n_{0,0})^*u^n_{0,0}$ is computed by
\[\varphi_{O_2^+}((a^{nk})^*a^{nk})=\varphi_{O_2^+}((u^{nk}_{0,0})^* u^{nk}_{0,0} )=\frac{1}{nk+1}=\left \|\pi^n_{0,0}\right \|_{2k}^{2k}=\varphi_{SU(2)}(\left |\pi^n_{0,0}\right |^{2k})\]
for all $k\in \left \{0\right\}\cup \n$. In other words, all the moments of $(u^n_{0,0})^*u^n_{0,0}$ and $\left |\pi^n_{0,0}\right |^2$ coincide. Therefore, we can conclude that
\[\left \|u^n_{0,0} \right \|_p^p=\varphi_{O_2^+}((\left |u^n_{0,0}\right |^2)^{\frac{p}{2}})=\varphi_{SU(2)}((\left |\pi^n_{0,0}\right |^2)^{\frac{p}{2}})=\left\| \pi^n_{0,0}\right \|_{p}^p=\frac{1}{\frac{np}{2}+1}\]
for all $1\leq p<\infty$. Therefore,
\[\lim_{p\rightarrow 2}\frac{2p}{2-p}\log(\frac{\left \|\sqrt{n+1}u^n_{0,0}\right \|_{L^p(O_2^+)}}{\left \|\sqrt{n+1}\widehat{u^n_{0,0}}\right \|_{\ell^{p'}(\widehat{O_2^+})}})=\frac{n}{n+1}\leq 1.\]

\item (The case of $SU_q(2)$)

For all $n\geq 0$ and $0\leq j\leq n$, we have
\begin{align*}
&\lim_{p\rightarrow 2} \frac{2p}{2-p} \log(\frac{\left \|u^n_{n,j}\right \|_{L^p(SU_q(2))}}{\left \|\widehat{u^n_{n,j}}\right \|_{\ell^{p'}(\widehat{SU_q(2)})}})\\
&=\lim_{p\nearrow 2} \frac{2p}{2-p} (\log(\frac{\left \|u^n_{n,j}\right \|_{L^p(SU_q(2))}}{\left \|u^n_{n,j}\right \|_{L^{2}(SU_q(2))}})+\log(\frac{\left \|\widehat{u^n_{n,j}}\right \|_{\ell^2(\widehat{SU_q(2)})}}{\left \|\widehat{u^n_{n,j}}\right \|_{\ell^{p'}(\widehat{SU_q(2)})}}))
\end{align*}
and $\displaystyle  \log(\frac{\left \|u^n_{n,j}\right \|_{L^p(SU_q(2))}}{\left \|u^n_{n,j}\right \|_{L^{2}(SU_q(2))}})\leq 0$ for all $1\leq p\leq 2$. Thus, by Lemma \ref{lem1} (2),

\begin{align*}
\lim_{p\rightarrow 2} \frac{2p}{2-p} \log(\frac{\left \|u^n_{n,j}\right \|_{L^p(SU_q(2))}}{\left \|\widehat{u^n_{n,j}}\right \|_{\ell^{p'}(\widehat{SU_q(2)})}}) &\leq \lim_{p \nearrow 2}\frac{2p}{2-p}\log (\frac{\left \|\widehat{u^n_{n,j}}\right \|_{\ell^2(\widehat{SU_q(2)})}}{\left \|\widehat{u^n_{n,j}}\right \|_{\ell^{p'}(\widehat{SU_q(2)})}})\\
&=\lim_{p' \searrow 2}\frac{2p'}{p'-2}\log (\frac{\left \|\widehat{u^n_{n,j}}\right \|_{\ell^2(\widehat{SU_q(2)})}}{\left \|\widehat{u^n_{n,j}}\right \|_{\ell^{p'}(\widehat{SU_q(2)})}})\\
&=\lim_{p' \searrow 2}\frac{2p'}{p'-2}\log ((d_n(Q_n)_{n,n})^{\frac{p'-2}{2p'}})\\
&=\log(d_n q^n)\leq \log(\frac{1}{1-q^2}).
\end{align*}

\end{enumerate}
\end{proof}

\section{The divergence of uncertainty relations}

In principle, if $\g$ is not finite, the uncertainty relation should be divergent 
\[e^{\lim_{q\rightarrow 2}\frac{2q}{2-q}\log(\frac{\left \|f\right \|_{L^q(\g)}}{\left \|\widehat{f}\right\|_{\ell^{q'}(\widehat{\g})}})}>>1\]
at certain element $f\in \mathrm{Pol}(\g)$ by Corollary \ref{cor0} and the fact that the Fourier transform $\mathcal{F}:L^1(\g)\rightarrow \ell^{\infty}(\widehat{\g})$ is not bounded below. But finding explicit elements showing the divergence is worthy of independent attention.

Indeed, Corollary \ref{cor1} (2) and Examples \ref{ex3}, \ref{ex2}, \ref{ex1} provide us with explicit elements showing the divergence. However, we need a different idea for connected compact Lie groups, $O_2^+$ or $SU_q(2)$ in view of Section \ref{sec:neg}, and the main point of this Section is that the divergence of uncertainty relations is attained at certain (linear combinations of) characters for compact quantum groups mentioned above.

\begin{theorem}\label{thm4}
\begin{enumerate}
\item For $G$ a connected compact Lie group, we have
\begin{equation}
e^{H(\left |\sum_{\pi\in \mathrm{Irr}(G)}\frac{m_{\pi}}{\left \|m\right \|_2 }\chi_{\pi}\right |^2,\varphi_G)+H(\left |\sum_{\pi\in \mathrm{Irr}(G)}\frac{m_{\pi}}{\left \|m\right\|_2}\widehat{\chi_{\pi}}\right |^2,\widehat{\varphi}_{\widehat{G}})} \gtrsim \min_{\pi\in supp(m)}\frac{n_{\pi}^2}{m_{\pi}^2}
\end{equation}
for any non-zero finite integral sequence $m=(m_{\pi})_{\pi\in \mathrm{Irr}(G)}\subseteq \z$. Here, $\left \|m\right\|_2=\displaystyle (\sum_{\pi\in \mathrm{Irr}(G)}m_{\pi}^2)^{\frac{1}{2}}$ and $supp(m)=\left \{\pi\in \mathrm{Irr}(G):m_{\pi}\neq 0\right\}$.

\item Let $\g$ be the free orthogonal quantum group $O_N^+$ with $N\geq 2$. Then for any non-zero finite integral sequence $m=(m_{k})_{k\geq 0}\subseteq \z$, we have
\begin{equation}
e^{H(\left |\sum_{k\geq 0}\frac{m_{k}}{\left \|m\right \|_2 }\chi_{k}\right |^2,\varphi_{O_N^+})+H(\left |\sum_{k\geq 0}\frac{m_{k}}{\left \|m\right\|_2}\widehat{\chi_{k}}\right |^2,\widehat{\varphi}_{\widehat{O_N^+}})} \gtrsim \min_{k\in supp(m)}\frac{n_{k}^2}{m_k^2}.
\end{equation}

\end{enumerate}

\end{theorem}

\begin{proof}
\begin{enumerate}
\item By Corollary \ref{cor0}, it is enough to show that
\[\frac{\left \| \sum_{\pi\in \mathrm{Irr}(G)}m_{\pi}\chi_{\pi}\right \|_{L^1(G)}}{\left \| \sum_{\pi\in \mathrm{Irr}(G)}m_{\pi}\widehat{\chi_{\pi}}\right \|_{\ell^{\infty}(\widehat{G})}}\gtrsim  \min_{\pi\in supp(m)}\frac{n_{\pi}}{m_{\pi}}.\]

Indeed, by [Lemma, \cite{Pr75}], there exists a universal constant $c(G)$ such that
\[\frac{\left \| \sum_{\pi\in \mathrm{Irr}(G)}m_{\pi}\chi_{\pi}\right \|_{L^1(G)}}{\left \| \sum_{\pi\in \mathrm{Irr}(G)}m_{\pi}\widehat{\chi_{\pi}}\right \|_{\ell^{\infty}(\widehat{G})}}\geq \frac{c(G)}{\max_{\pi\in supp(m)}\frac{m_{\pi}}{n_{\pi}}}=c(G)\min_{\pi\in supp(m)}\frac{n_{\pi}}{m_{\pi}}\]
\item Through [Lemma 4.7, \cite{Yo18}] and Theorem \ref{thm4} (1),  we are able to obtain the conclusion.

\end{enumerate}
\end{proof}

Note that Theorem \ref{thm4} and (\ref{eq6}) tells us that the uncertainty relations diverge at characters since
\begin{equation}\label{ineq12}
e^{H(\left |\chi_{\alpha}\right |^2,\varphi_{\g})+H(\left |\widehat{\chi_{\alpha}}\right |^2,\widehat{\varphi}_{\widehat{\g}})}\sim n_{\alpha}^2~\mathrm{for~all~}\alpha\in \mathrm{Irr}(\g)
\end{equation}
if $\g$ is a connected compact Lie group or a free orthogonal quantum group.

In contrast, this divergence at characters does not appear in the case of the quantum $SU(2)$ group $SU_q(2)$.

\begin{proposition}\label{prop3}
For the quantum $SU(2)$ group $SU_q(2)$ with $0<q<1$, we have
\begin{equation} 
\sup_{n\in \n} e^{\lim_{p\rightarrow 2} \frac{2p}{2-p} \log(\frac{\left \|\chi_n \right \|_{L^p(SU_q(2))}}{\left \|\widehat{\chi_n}\right \|_{\ell^{p'}(\widehat{SU_q(2)})}})} <\infty.
\end{equation}

\end{proposition}
\begin{proof}
For each $n\in \n$, note that $\displaystyle \widehat{\chi_n}(k)=\frac{\delta_{n,k}}{d_n}Q_{n}^{-1}\in M_{n+1}\subseteq \ell^{\infty}(\widehat{SU_q(2)})$. Therefore,

\begin{align*}
\lim_{p\nearrow 2} \frac{2p}{2-p} \log(\frac{\left \|\chi_n\right \|_{L^p(SU_q(2))}}{\left \|\widehat{\chi_n}\right \|_{\ell^{p'}(\widehat{SU_q(2)})}})&\leq \lim_{p \nearrow 2}\frac{2p}{2-p}\log (\frac{1}{\left \|\widehat{\chi_n }\right \|_{\ell^{p'}(\widehat{SU_q(2)})}})\\
&=\lim_{p' \searrow 2}\frac{2p'}{2-p'}\log ((d_n^{1-p'}\mathrm{tr}(Q_n^{1-p'}))^{\frac{1}{p'}})\\
&=-2 \sum_{j=0}^n \frac{1}{d_nq^{-n+2j}}\log(\frac{1}{d_nq^{-n+2j}})=2H(x_n),
\end{align*}
where $x_n=\displaystyle (\frac{1}{d_nq^{-n+2j}})_{j=0}^n$. Now, let us show that $\left \{ e^{2H(x_n)}:n\in \n \right\}$ is uniformly bounded. Recall that $\displaystyle d_n=\frac{q^{-n}(1-q^{2n+2})}{1-q^2}$ for all $n\in \n$. Then
\begin{align*}
H(x_n)&=\sum_{j=0}^n \frac{1}{d_n q^{-n+2j}}\log(d_nq^{-n+2j})\\
&=\log(d_nq^{-n})+\frac{1}{d_n}\sum_{j=0}^n \frac{\log(q^{2j})}{q^{-n+2j}}\\
&=\log(d_nq^{-n})+\frac{2q^n\log(q)}{d_n}\sum_{j=0}^n jq^{-2j}\\
&=\log(d_nq^{-n})+\frac{2q^{2}\log(q)}{1-q^{2n+2}}(nq^{-2}-\frac{1-q^{2n}}{1-q^2})\\
&=-2n\log(q)+\log(\frac{1-q^{2n+2}}{1-q^2})+\frac{2q^{2}\log(q)}{1-q^{2n+2}}(nq^{-2}-\frac{1-q^{2n}}{1-q^2})\\
&=\frac{2nq^{2n+2}}{1-q^{2n+2}}\log(q)+\log(\frac{1-q^{2n+2}}{1-q^2}) - \frac{2q^2(1-q^{2n})\log(q)}{(1-q^2)(1-q^{2n+2})}.
\end{align*}

Now, it is sufficient to see that
\[\lim_{n\rightarrow \infty}H(x_n)=0+\log(\frac{1}{1-q^2})+\frac{2q^2\log(q^{-1})}{1-q^2}<\infty.\]

\end{proof}

\begin{remark}
Let $G$ be a connected compact Lie group. Then the estimates (\ref{ineq12}) can be also explained from the facts that
\begin{equation}
\inf_{\pi\in \mathrm{Irr}(G)}H(\left |\chi_{\pi}\right |^2,\varphi_G)>-\infty~\mathrm{and~}H(\left |\widehat{\chi_{\pi}}\right |^2,\widehat{\varphi}_{\widehat{G}})= \log(n_{\pi}^2),
\end{equation}
The former can be proved by Corollary \ref{cor0} and [Theorem 5.4, \cite{Do79}].

\end{remark}

Although connected semisimple compact Lie groups do not admit infinite (local) central $\Lambda(4)$-sets (\cite{Ce72} and [Corollary 7, \cite{GiSoTr82}]), a strong contrast holds for $SU_q(2)$. More precisely, for $E=\left \{m_k=\frac{k(k+1)}{2}:k\in \n \right\}\subseteq \left \{0\right\}\cup \n$, there exists a universal constant $K=K(q,E)$ such that
\begin{equation}\label{ineq7}
\left \|\sum_{n\in E} a_n \chi_n \right \|_{L^4(SU_q(2))}\leq K \left \|\sum_{n\in E}a_n \chi_n \right \|_{L^2(SU_q(2))}
\end{equation}
for any finitely supported sequence $(a_n)_{n\in E}\subseteq \Comp$ [Proposition 5.16, \cite{Wa17}]. Through a similar proof of Theorem \ref{thm3}, we can apply this lacunarity result to detect the divergence of uncertainty relations in $SU_q(2)$.

\begin{corollary}\label{thm5}
Set $E=\left \{m_k=\frac{k(k+1)}{2}:k\in \n \right\}\subseteq \left \{0\right\}\cup \n$, Then for any $f=\displaystyle \sum_{n\in E}a_n\chi_n\in \mathrm{Pol}_E(SU_q(2))$ with $\left \|f\right\|_{L^2(SU_q(2))}=1$, we have
\begin{equation}
 e^{\lim_{p\rightarrow 2} \frac{2p}{2-p} \log(\frac{\left \|f \right \|_{L^p(SU_q(2))}}{\left \|\widehat{f}\right \|_{\ell^{p'}(\widehat{SU_q(2)})}})}\gtrsim e^{H( (\left |a_n \right |^2)_{n\in E})}
\end{equation}

\end{corollary}

\begin{proof}

By Corollary \ref{cor0}, for all $2<p<4$ and $f=\displaystyle \sum_{n\in E}a_n\chi_n\in \mathrm{Pol}_E(SU_q(2))$ with $\left \|f\right\|_2=1$, we have
\[\left \|f\right\|_p \leq K^{\frac{2(p-2)}{p}}\left \|f\right\|_2\mathrm{~and~} \left \|\widehat{f}(\alpha)\right \|_{\infty}^{2-p'}\left \|\widehat{f}(\alpha)Q_{\alpha}^{\frac{1}{p'}}\right\|_{S^{p'}_{n+1}}^{p'}\geq \left \|\widehat{f}(\alpha)Q_{\alpha}^{\frac{1}{2}}\right\|_{HS}^2\]
for all $\alpha\in \mathrm{supp}(\widehat{f})$. Therefore, for all $2<p<4$, we have 
\begin{align*}
\frac{2p}{p-2}\log(\frac{\left \|\widehat{f}\right \|_{p'}}{\left \|f\right \|_p})&=\frac{2p}{p-2}(\log (\frac{\left \|\widehat{f}\right \|_{p'}}{\left \|\widehat{f}\right \|_2})+\log (\frac{\left \|f\right \|_{2}}{\left \|f\right \|_p}))\\
& \geq  \frac{2p}{p-2}( \log((\sum_{n\in E}d_n^{1-p'}\left |a_n\right |^{p'} \mathrm{tr}(Q_n^{1-p'}))^{\frac{1}{p'}})-\frac{2(p-2)}{p}\log(K)) \\
&=\frac{2}{2-p'}\log(\sum_{n\in E}d_n^{1-p'}\left |a_n\right |^{p'} \mathrm{tr}(Q_n^{1-p'}))-4\log(K).
\end{align*}

Therefore,
\begin{align*}
&K^4e^{\lim_{p\rightarrow 2}\frac{2p}{2-p}\log(\frac{\left \|f\right \|_p}{\left \|\widehat{f}\right \|_{p'}}) }\\
&\geq e^{-2(\sum_{n\in E}\left |a_n \right |^2 \log(\left |a_n \right |)+\sum_{n\in E}\left |a_n \right |^2 \sum_{j=0}^n \frac{1}{d_n(Q_n)_{j,j}}\log(\frac{1}{d_n(Q_n)_{j,j}}))}\\
&=e^{ H((\left | a_n\right |^2)_{n\in E})+2 \sum_{n\in E}\left |a_n\right |^2 H(x_n)} \geq e^{H((\left |a_n\right |^2)_{n\in E}) },
\end{align*}
where $x_n$ is the sequence defined in the proof of Proposition \ref{prop3}.

\end{proof}

\begin{remark}
In the case of $SU_q(2)$, we can pick an explicit sequence of matrix elements $(u_{0,n})_{n\geq 0}$ showing the divergence of the uncertainty relations. Indeed, we have 
\begin{equation}
e^{\lim_{p\rightarrow 2}\frac{2p}{2-p}\log(\frac{\left \|\sqrt{q^{-n}d_n} u^n_{0,n}\right \|_{L^{p}(SU_q(2))}}{\left \|\sqrt{q^{-n}d_n}\widehat{u^n_{0,n}}\right \|_{\ell^{p'}(\widehat{SU_q(2)})}})} \sim d_n^2 ~\mathrm{for~all~}n\in \n.
\end{equation}
by Corollary \ref{cor0}, Proposition \ref{prop2} and the fact that $\displaystyle \frac{\left \|\widehat{u^n_{0,n}}\right \|_{\ell^1(\widehat{SU_q(2)})}}{\left \|u^n_{0,n}\right \|_{L^{\infty}(SU_q(2))}}=\frac{1}{q^n}$ for all $n\in \n$ [Theorem 5.4, \cite{Ko89}].
\end{remark}

\bibliographystyle{alpha}
\bibliography{Uncertainty}

\end{document}